\documentclass[sigconf, nonacm]{acmart}
\usepackage{siunitx}
\usepackage{graphicx}
\usepackage{float, amsthm}
\usepackage{caption}
\usepackage{subcaption,enumitem}
\usepackage{comment}
\usepackage{mathtools, multirow}
\usepackage{color}

\usepackage[table]{xcolor}

\usepackage[ruled,vlined,linesnumbered]{algorithm2e}
\usepackage{amsmath}  
\usepackage{amssymb}

\SetAlFnt{\tiny\sffamily}
\SetAlCapFnt{\tiny\sffamily\bfseries}
\SetAlCapNameFnt{\tiny\sffamily}
\SetCommentSty{\color{gray}\sffamily}

\AtBeginDocument{%
  \providecommand\BibTeX{{%
    \normalfont B\kern-0.5em{\scshape i\kern-0.25em b}\kern-0.8em\TeX}}}

\renewcommand\footnotetextcopyrightpermission[1]{} 
\begin{document}

\title{Sense Less, Infer More: Agentic Multimodal Transformers for Edge Medical Intelligence}

\author{Chengwei Zhou$^{1,*}$, Zhaoyan Jia$^{2,*}$, Haotian Yu$^{1}$, Xuming Chen$^{1}$, Brandon Lee$^{1}$\\
Christopher Pulliam$^{1}$, Steve Majerus$^{1}$, Massoud Pedram$^{2}$, and Gourav Datta$^{1}$ \vspace{0.3em}\\
\small $^1$Case Western Reserve University, USA \\
$^2$University of Southern California, USA \\
$^*$Equal contributions
\vspace{-0.3em}
\\}






\begin{abstract}

Edge-based multimodal medical monitoring requires models that balance diagnostic accuracy with severe energy constraints. Continuous acquisition of ECG, PPG, EMG, and IMU streams rapidly drains wearable batteries, often limiting operation to under 10 hours, while existing systems overlook the high temporal redundancy present in physiological signals. We introduce \textbf{Adaptive Multimodal Intelligence (AMI)}, an end-to-end framework that jointly learns \emph{when} to sense and \emph{how} to infer. AMI integrates three components: (1) a lightweight \textit{Agentic Modality Controller} that uses differentiable Gumbel--Sigmoid gating to dynamically select active sensors based on model confidence and task relevance; (2) a \textit{Learned Sigma--Delta Sensing} module that applies patch-wise Delta--Sigma operations with learnable thresholds to skip temporally redundant samples; and (3) a \textit{Foundation-backed Multimodal Prediction Model} built on unimodal foundation encoders and a cross-modal transformer with temporal context, enabling robust fusion even under gated or missing inputs. These components are trained jointly via a multi-objective loss combining classification accuracy, sparsity regularization, cross-modal alignment, and predictive coding. AMI is hardware-aware, supporting dynamic computation graphs and masked operations, leading to real energy and latency savings. Across MHEALTH, HMC Sleep, and WESAD datasets, it reduces sensor usage by $48.8\%$ while improving the state-of-the-art accuracy by $1.9\%$ on average. 
Theoretically, AMI achieves $\mathcal{O}(k^*/\epsilon^2)$ sample complexity with logarithmic convergence, improving upon the $\mathcal{O}(M/\epsilon^2)$ requirement of decoupled sensing.

\end{abstract}



\vspace{-8mm}



\maketitle
\pagestyle{plain}

\vspace{-2mm}
\section{Introduction}

Edge-based medical monitoring systems face a fundamental challenge: multimodal sensors (ECG, PPG, EMG, IMU, and even audio or respiration) provide rich physiological data, but dramatically increase energy consumption on battery-powered devices \cite{shimmer,ads1292r,max30101}. Fig.~\ref{fig:battery} illustrates this issue: using power values reported in sensor datasheets (e.g., 0.3–1 mW IMU, 1–5 mW ECG, 6–15 mW EMG, 4–10 mW PPG), a wearable with a 300 mWh battery can support each sensor alone for hundreds of hours, yet combining them reduces runtime to under 10 hours. This mismatch between multimodal sensing and limited battery capacity severely restricts long-term, continuous monitoring scenarios such as sleep staging, cardiac surveillance, or stress detection \cite{mhealth_319,schmidt2018introducing}. 
This challenge has become even more severe with the rise of on-device multimodal machine learning (ML) technologies. Modern architectures, including cross-modal Transformers, multimodal sleep staging models, and physiological foundation models \cite{pillai2024papagei,abbaspourazad2024largescale,thapa,fang2024promoting}, achieve state-of-the-art (SOTA) accuracy by aligning and fusing heterogeneous biosignals, but their attention layers, tensor operations, and memory access patterns substantially exceed the energy cost of sensing itself \cite{s24196322}. Hence, multimodal sensing combined with continuous on-device ML inference can drain a wearable battery in under an hour, making current approaches incompatible with even continuous monitoring. 

\begin{figure}[!t]
    \centering
    \vspace{-1mm}\includegraphics[width=0.9\linewidth]{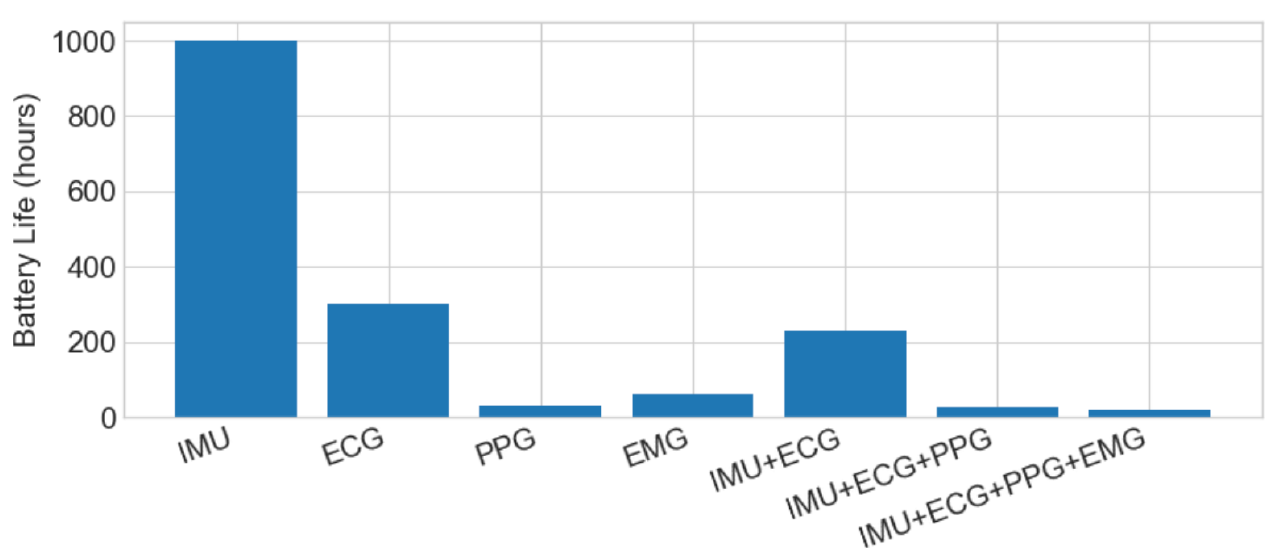}
    \vspace{-5mm}
    \caption{\small Illustrative battery life for uni-modal vs. multi-modal sensing (300 mWh battery). Low-power IMU-only monitoring can last for many days, but adding ECG, PPG, and EMG rapidly reduces battery life to tens of hours.}
    \vspace{-7mm}
    \label{fig:battery}
\end{figure}

The core insight of this work is that \textit{sensing decisions—both which sensors to activate and when to sample them—must be learned together with the inference task itself}. By making the multimodal model agentic, allowing it to dynamically control its own sensing policy while simultaneously learning to fuse and exploit discriminative information across modalities, the system can dramatically reduce energy consumption without compromising diagnostic accuracy. This requires solving \textit{three interconnected challenges}:

\noindent
\textbf{1. Joint optimization of sensing and inference.}
Prior works treat sensor selection and model inference separately, relying on heuristics or independently trained selection policies. This disconnect leads to poor decisions, activating sensors when unnecessary or disabling them when crucial. A unified approach is needed where sensing is guided by the model’s confidence and task demands.

\noindent
\textbf{2. Leveraging temporal redundancy.}
Biosignals are highly correlated over time (e.g., stable heart rate, persistent motion patterns), yet existing multimodal models process every window independently. 
Integrating temporal reuse into end-to-end training is key to reducing sampling and computation.

\noindent
\textbf{3. Hardware-aware adaptation.}
Energy-constrained edge devices require models that adjust computation to available resources, support variable input sizes, and guarantee predictable runtime behavior. Current multimodal architectures assume fixed computational graphs and cannot dynamically scale with hardware budgets.

To address these challenges, we introduce \textbf{Adaptive Multimodal Intelligence (AMI)}, an end-to-end trainable framework (see Fig. ~\ref{fig:framework}) that unifies sensing decisions and inference within a single optimization process. AMI encompasses three key contributions:
\vspace{-4mm}
\begin{itemize}[leftmargin=8pt]
\item \textbf{A unified architecture for joint sensing and inference optimization.} We use lightweight unimodal foundation models as modality-specific encoders, providing rich features from large-scale pretraining. The agentic controller jointly learns \emph{when} to sense each modality and \emph{how} to fuse the unimodal features through cross-modal attention, with causal temporal layers capturing context across signals. We jointly train sensor selection and multimodal fusion with the prediction task using a novel multi-objective loss that combines classification accuracy, sparsity regularization, cross-modal alignment, and predictive coding. This directs sensing effort where it is most needed, improving accuracy and efficiency while leveraging the robustness of foundation-model representations.

\item \textbf{Temporal-aware sensing through learned Sigma-Delta ($\Sigma$-$\Delta$) modulation.} We develop a differentiable $\Sigma$-$\Delta$ Sensing mechanism that identifies and skips temporally redundant samples. Unlike traditional $\Sigma$-$\Delta$ converters that use fixed thresholds, our approach learns modality-specific change thresholds that balance information preservation with sampling reduction. 

\item \textbf{Hardware-efficient implementation with dynamic computation graphs.} Through careful co-design with TensorRT optimization, we design our architecture to support variable-length inputs and masked computation, enabling true energy savings on resource-constrained edge devices.
\end{itemize}
\vspace{-0.6mm}
We evaluate AMI on three publicly available multimodal biomedical datasets: MHEALTH (activity recognition), HMC Sleep (sleep staging), and WESAD (stress detection). AMI reduces latency and energy consumption {by 31.9\% and 24.8\% respectively} (directly translating to extended battery life in continuous monitoring settings) on average, while simultaneously surpassing the SOTA accuracy by 1.4\%.
Beyond empirical results, we prove that our joint optimization achieves: (1) factor $M/k^*$ fewer samples through selective sensing, and (2) logarithmic $\mathcal{O}(\log(1/\epsilon))$ convergence versus $\mathcal{O}(M/\epsilon^2)$ for decoupled methods.

AMI supports not only efficient state prediction from biosignals but also the downstream step in closed-loop systems: converting predicted states into actions. In many neuro-stimulation and intervention settings, this state-to-action mapping is governed by simple, validated clinical rules (e.g., trigger stimulation when a pathological state is detected). By delivering accurate, low-latency state estimates under reduced sensing budgets, our model strengthens the reliability of these rule-based controllers. Although our learning focuses on sensing and inference, it completes the active-sensing loop, enabling wearable and implantable systems that both identify patient state efficiently and initiate timely therapeutics.

\begin{figure}[!t]
    \centering
\includegraphics[width=0.95\linewidth]{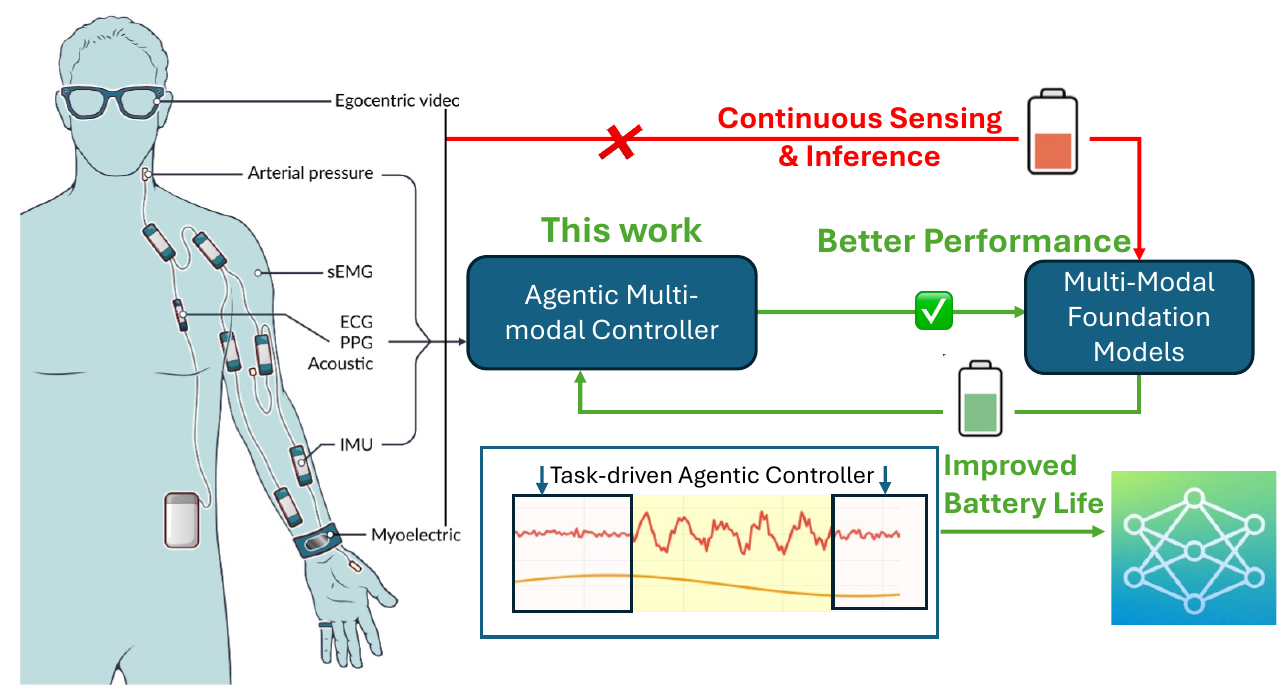}
\vspace{-3mm}
    \caption{\small Our Unified Agentic Multimodal Sensing and Inference Framework for Efficient, High-Accuracy Biomedical AI.}
    \vspace{-6mm}
    \label{fig:framework}
\end{figure}

\vspace{-2mm}
\section{Related Work}

\textbf{Multimodal Learning and Adaptive Sensing}: Multimodal deep learning has driven major progress in biomedical monitoring. Early CNN-based fusion of ECG, EEG, and motion signals~\cite{faust2018deep} improved robustness, while transformer architectures~\cite{zhang2022transformer, multimodal_fusion} enabled more expressive fusion mechanisms beyond simple concatenation. More recently, medical foundation models such as PhysioOmni~\cite{jiang2025towards}, SleepMG~\cite{ma2024sleepmg}, and medical vision--language adaptations~\cite{wang2023medical} have demonstrated strong generalization across tasks and modalities. However, these models rely on dense, continuous sampling from all sensors and incur heavy computational and energy costs. In parallel, sensor selection has been explored through sparsity methods~\cite{tibshirani1996regression, zou2005regularization}, RL-based policies~\cite{sharma2020sensor}, contextual bandits~\cite{bouneffouf2020contextual}, differentiable gating~\cite{jang2017categorical}, information-theoretic selection~\cite{zhao2021information}, Bayesian optimization~\cite{garnett2020bayesian}, meta-learning~\cite{finn2017model}, and neural architecture search (NAS)~\cite{elsken2019neural}. Yet these techniques operate \emph{outside} the prediction model and therefore cannot jointly optimize sensing, multimodal fusion, and task accuracy, limiting their practical effectiveness for energy-constrained, real-time medical devices.

\noindent
\textbf{Temporal Redundancy Reduction and Hardware Optimization}: Temporal redundancy has long been exploited in efficient sensing. Examples include $\Sigma$-$\Delta$ modulation~\cite{schreier2005understanding}, neuromorphic event sensors~\cite{liu2010neuromorphic}, compressed sensing~\cite{candes2006robust}, predictive coding~\cite{rao1999predictive}, next-frame prediction~\cite{lotter2017deep}, change detection~\cite{aminikhanghahi2017survey}, asynchronous processing~\cite{manohar2015asynchronous}, Skip-RNNs~\cite{campos2017skip}, and delta encoding~\cite{chen2010compressed}. Yet these methods mostly target \emph{single} modalities and are rarely paired with multimodal deep learning or optimized toward task-specific objectives. On the hardware side, edge-efficient inference has benefited from pruning, quantization, and distillation~\cite{han2016deep, jacob2018quantization, hinton2015distilling}, hardware-aware NAS~\cite{wu2019fbnet}, dynamic/conditional computation~\cite{han2021dynamic, bengio2013estimating}, early exits~\cite{teerapittayanon2016branchynet}, accelerator-level optimizations~\cite{nvidia2023tensorrt}, biomedical-specific hardware~\cite{lee2022biomedical}, in-memory computing~\cite{sebastian2020memory}, approximate computing~\cite{mittal2016survey}, and software–hardware co-design~\cite{9328612}. However, these techniques optimize \emph{fixed} computational graphs and cannot adapt to varying sampling rates or input dimensionality on the fly as signal content evolves.


\vspace{-3mm}

\section{Methodology}

In this section, we formalize the AMI framework (see Fig. ~\ref{fig:architecture} and Algorithm ~\ref{alg:empm}) that integrates
three key components:
(1) a \textbf{Foundation-backed Multimodal Prediction Model} (FMPM),
(2) an \textbf{Agentic Modality Controller} (AMC) responsible for modality-level gating, and
(3) a \textbf{Sigma–Delta Sensing Module} for patch-level temporal gating.
We subsequently introduce our training objective, and provide theoretical support for jointly optimizing masked sensing and inference for multimodal bio-signals.

\vspace{-2.5mm}
\subsection{Multimodal Prediction Model}

The proposed Foundation-backed Multimodal Prediction Model (FMPM) integrates foundation model-embedding encoder, cross-modal fusion, and temporal context into a unified Transformer-based framework. It aims to improve robustness under modality dropout and enhance inter-modality information flow.

\noindent\textbf{Foundation Modality Encoders}: Instead of learning each modality tokenizer from scratch, we employ SOTA lightweight \emph{foundation modality encoders} $f_{\text{FM}}^{(m)}$ pretrained on large-scale unimodal datasets. 
Specifically, we deployed ECG-FM ~\cite{mckeen2025ecg} for ECG signals and PaPaGei~\cite{pillai2024papagei} for PPG signals.
Each encoder converts a raw signal $x^{(m)}$ into token embeddings as $T^{(m)} = f_{\text{FM}}^{(m)}(x^{(m)}) \in \mathbb{R}^{B \times L \times D}$. These pretrained encoders provide domain-aware feature extraction, accelerating convergence and improving downstream generalization. Depending on compute constraints, the encoders can be fine-tuned or frozen during training.

\noindent\textbf{Cross-Modal Attention Fusion}: To enable fine-grained information exchange among heterogeneous modalities, we design an all-to-all cross-attention fusion module. 
Given $M$ modality-specific token sequences $\{T^{(1)}, T^{(2)}, \dots, T^{(M)}\}$, each modality $i$ attends to the token representations of all other modalities:
\begin{equation}
    \texttt{KV}^{(i)} = \texttt{concat}\big(\{T^{(j)}\}_{j \neq i}\big), \quad
    \hat{T}^{(i)} = \texttt{CrossAttn}\big(T^{(i)}, \texttt{KV}^{(i)}\big),
\end{equation}
where $\texttt{CrossAttn}(\cdot)$ denotes a pre-norm multi-head attention followed by a feed-forward block with residual connections. 
The attended representations are concatenated to form the fused multimodal feature as $H_{\text{fused}} = \texttt{concat}\big(\hat{T}^{(1)}, \hat{T}^{(2)}, \dots, \hat{T}^{(M)}\big){\in}\mathbb{R}^{B \times (M L) \times D}$. This design allows each modality to dynamically query relevant features from other modalities, improving representational consistency and information sharing.

\noindent\textbf{Cross-Attention-based Temporal Encoding}: To incorporate temporal priors and reduce degradation under partial sensor gating, we add a cross-attention-based contextual encoder that integrates recent history. Let $\bar{H}_{t-K:t-1}$ be a compact memory of the past $K$ fused representations (mean-aggregated). At time $t$, the current features $\mathbf{R}_{t}^{L}$ attend to this memory via $H'_t = \texttt{CrossAttn}\big(H_{t}^{L}, \bar{H}_{t-K:t-1}\big)$, which enables the model to capture temporal coherence and long-range dependencies without a full recurrent state. After cross-modal fusion and contextual encoding, the fused representation is processed by a lightweight Transformer fusion block, followed a learned \texttt{[CLS]} token and a linear classification head.

\begin{figure}[!t]
  \centering
  \includegraphics[trim={200 230 220 200},clip,width=1\linewidth]{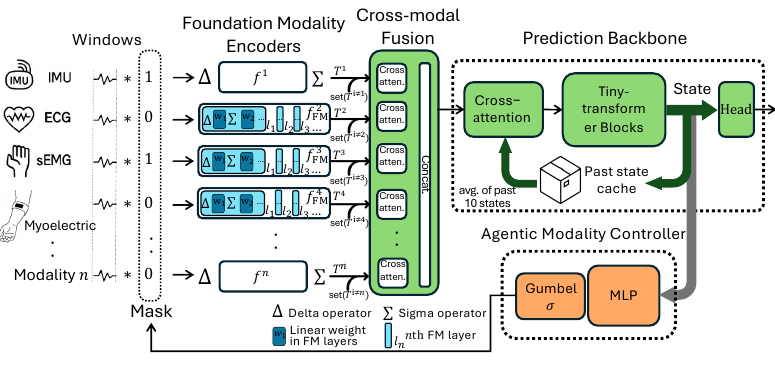}
  \vspace{-10mm}
    \caption{\small Proposed AMI architecture. The $\Sigma$-$\Delta$ module performs patch-wise computation and adaptive skipping. Features are fused via early cross-attention and processed by a Tiny-Transformer backbone with context reasoning. A Gumbel-MLP agent dynamically produces per-modality masks for the next window.}
    \label{fig:architecture}
  \vspace{-5mm}
\end{figure}



\vspace{-3mm}
\subsection{Agentic Modality Controller}

The Agentic Modality Controller (AMC) is a lightweight module that adaptively selects which modalities to process in the next time window. It takes the current FMPM representation and outputs per-modality gating decisions, and is trained jointly with the FMPM to balance task accuracy and sensing cost.


\noindent
\textbf{Feature aggregation and gate logits}: We first aggregate patch-level representations \(\mathbf{R}\in\mathbb{R}^{B\times (ML)\times D}\) into modality-level features, where $M$ is total modality number and $L$ is the number of patches per modality. Reshape and average over patches, and then we flatten the per-modality means to get \(\bar{\mathbf{R}} \in \mathbb{R}^{B\times M\times D}\), and feed them into a small MLP gate feature extractor \(g_{\phi}(\cdot)\): $\mathbf{\ell} = g_{\phi}\big(\bar{\mathbf{R}}\big)\in\mathbb{R}^{B\times M}$, where \(\mathbf{\ell}\) are pre-sigmoid gate logits with one logit per modality.

\noindent
\textbf{Gumbel--Sigmoid sampling and straight-through estimator}: To enable exploration during training we apply a Gumbel perturbation to logits and compute a continuous relaxation:
\vspace{-1.5mm}
\begin{align}
    u &\sim \texttt{Uniform}(0,1),\qquad g=-\log\big(-\log(u+\varepsilon)+\varepsilon\big),\\
    \tilde{\ell} &= (\ell + g)/\tau,\qquad
    p_{\texttt{soft}}=\sigma(\tilde{\ell}).
\end{align}
When $\tau$ is small the relaxation concentrates near $\{0,1\}$; at inference we omit noise and set $p_{\texttt{soft}}=\sigma(\ell)$. We use binary gating in the forward pass but propagate gradients through continuous relaxation. 
Formally, the binary decision for each modality is $p_{\texttt{hard}} = \mathbf{1}\{p_{\texttt{soft}}>0.5\}\in\{0,1\}^{B\times M}$, and the AMC applies $p_{\texttt{hard}}$ to mask sensors when computing the next-window inputs.
During backpropagation we treat the forward output as if it were the continuous $p_{\texttt{soft}}$ when computing gradients. Equivalently, for any scalar loss $\mathcal{L}$ that depends on the forward gating output $p_{\texttt{forward}}$, we use the chain rule approximation: $\frac{\partial \mathcal{L}}{\partial \ell}\approx
\frac{\partial \mathcal{L}}{\partial p_{\texttt{forward}}}\cdot
\frac{\partial p_{\texttt{soft}}}{\partial \ell}$.

The AMC adds a small feedback latency, but the combined delay of the FMPM and AMC (tens of milliseconds) is far shorter than the 1s input window we process  for continuous monitoring (100 samples at 100 Hz). Thus, sensor-mask decisions always arrive before the next window, and \textit{the method introduces no throughput or timing overhead} compared to the baseline.


\vspace{-3mm}
\subsection{Sigma-Delta Sensing}
To further improve efficiency and reduce redundant sampling and sensing, we propose a \textbf{Sigma--Delta ($\Sigma$-$\Delta$) Sensing} mechanism that exploits temporal redundancy across consecutive patches within each modality. Inspired by $\Sigma$-$\Delta$ modulation~\cite{o2016sigma}, the module operates during tokenization and adaptively determines whether the current patch contains sufficiently informative temporal change. Only patches exhibiting meaningful variations are actively tokenized, while stable patches reuse previously emitted tokens. 

\begin{figure}[t]
    \centering
    \includegraphics[trim={280 290 270 260},clip,width=1\linewidth]{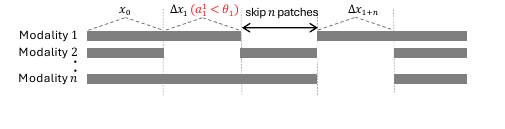}
    \vspace{-14mm}
    \caption{Thresholding and skipping in Sigma--Delta Sensing.}
    \label{fig:skip}
    \vspace{-5mm}
\end{figure}

\noindent\textbf{Patch-wise $\Sigma$-$\Delta$ Operation}: Let $\mathbf{x}^{(m)} \in \mathbb{R}^{C_m \times T}$ denote the raw input of modality $m$. We partition it into $L = T/P$ non-overlapping patches of size $P$ to get $\mathbf{x}_l^{(m)} \in \mathbb{R}^{C_m \times P}$. To capture temporal change within a window, we compute the patch-wise difference: $\Delta \mathbf{x}_l^{(m)} = \mathbf{x}_l^{(m)} - \mathbf{x}_{l-1}^{(m)}$, where $\ell\in[1,L]$. For convolution-based tokenizers, Conv1D only processes the difference patch $\Delta\mathbf{x}_l^{(m)}$. For foundation-model tokenizers, the $\Delta$ is applied prior to the initial linear projection (normally, it is a convolutional layer, the patch size is equal to kernel size), while subsequent layers operate on the accumulated ($\Sigma$) output signal.

\noindent\textbf{Adaptive Thresholding and Skip Policy}: As shown in Fig. ~\ref{fig:skip}, each modality maintains a sensitivity threshold $\theta_m$ and a skip horizon $k_{\text{skip}}$. To quantify patch activity, we compute the normalized magnitude of temporal change as $a_l^{(m)} 
    = \frac{1}{C_m P} 
    \sum_{c,p} \left| \Delta \mathbf{x}_{l,c,p}^{(m)} \right|$. A patch is considered \emph{inactive} (stable) if $a_l^{(m)} < \theta_m$, in which case tokenization is skipped for up to $k_{\text{skip}}$ consecutive patches. Otherwise, the patch is active. 
    The emitted token sequence is updated recursively:
    \vspace{-1mm}
\begin{equation}
    \tilde{\mathbf{T}}_l^{(m)} =
    \begin{cases}
        \texttt{Tokenizer}\bigl(\Delta \mathbf{x}_l^{(m)}\bigr), 
        & \text{if patch } l \text{ is active}, \\[6pt]
        \tilde{\mathbf{T}}_{l-1}^{(m)}, 
        & \text{if patch } l \text{ is skipped}.
    \end{cases}
\end{equation}

This patch-wise, intra-window $\Sigma$-$\Delta$ design needs only a single patch buffer per modality and adds negligible memory overhead. It enables true sensor duty-cycling: skipped patches incur no sampling or tokenization cost, significantly reducing the active sampling rate of the sensors 
with minimal impact on accuracy.

\vspace{-0.5mm}
\begin{algorithm}
\scriptsize
\raggedright
\caption{\small Pseudo code of AMI: Combining FMPM, AMC, $\Sigma$-$\Delta$ Sensing.}
\label{alg:empm}
\DontPrintSemicolon
\SetKwInput{KwIn}{Input}
\SetKwInput{KwOut}{Output}
\KwIn{Multimodal signal $X\in\mathbb{R}^{B\times M\times T}$; optional history $H_{\texttt{hist}}$}
\KwOut{Predicted logits $\mathbf{y}$ and current representation features $H$}
\BlankLine
\For{$m=1$ \KwTo $M$}{
    $x_m \leftarrow X[:,m,:]$\;
    $T^{(m)} \leftarrow \Sigma f_{}^{(m)}(\Delta x_m)$ 
}

\BlankLine

\ForEach{modality $i$}{
    $\texttt{KV}^{(i)} \leftarrow \texttt{concat}\!\bigl(\{T^{(j)}\}_{j\neq i}\bigr)$\;
    $\hat{T}^{(i)} \leftarrow \texttt{CrossAttn}\!\bigl(T^{(i)},\texttt{KV}^{(i)}\bigr)$\;
}
$H_{\texttt{fused}} \leftarrow \texttt{concat}\!\bigl(\hat{T}^{(1)},\dots,\hat{T}^{(M)}\bigr)$\  

\BlankLine

$H_{\texttt{fused}} \leftarrow \texttt{PositionalEncoding}(H_{\texttt{fused}})$\;

$H_{\texttt{fused}} \leftarrow \texttt{CrossAttn}(H_{\texttt{fused}},H_{\texttt{hist}})$\ 

$H \leftarrow \texttt{Tiny-Transformer}(H_{\texttt{fused}})$\ 

$\mathbf{y} \leftarrow \texttt{Classifier}\!\bigl(H_{cls\_token}\bigr)$\;


$R\leftarrow H,   \ell\leftarrow g_\phi(\text{mean}(R))$\;

\If{train}{
  add Gumbel noise, $p_{\text{soft}}\leftarrow\sigma((\ell+g)/\tau)$; $p_{\text{hard}}\leftarrow\mathbf{1}\{p_{\text{soft}}>0.5\}$ (ST)
}\Else{
  $p_{\text{soft}}\leftarrow\sigma(\ell)$; $p_{\text{hard}}\leftarrow\mathbf{1}\{p_{\text{soft}}>0.5\}$
}

\Return{$\mathbf{y},H, p_{\text{soft}}, p_{\text{hard}}$}\;

\end{algorithm}
\vspace{-2mm}

\vspace{-1mm}
\subsection{Training Objective}
\begin{figure}
    \centering
    \includegraphics[trim={255 275 230 210},clip,width=1\linewidth]{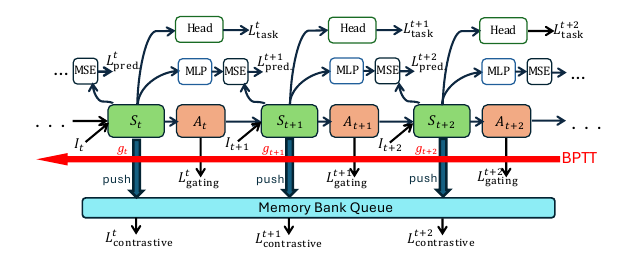}
    \vspace{-9mm}
    \caption{
    \small Training pipeline with unrolled timesteps optimized via BPTT. 
    Each fused state $S_t$ produces a prediction loss, gating loss, predictive coding loss, and a contrastive alignment loss computed against a memory bank. 
    The controller’s gating actions $A_t$ influence future observations, and all losses jointly update the model through temporal backpropagation.
    }
    \label{fig:pipeline}
    \vspace{-5mm}
\end{figure}

To jointly optimize predictive accuracy, sensing efficiency, and temporal consistency, AMI combines complementary objectives in a unified loss. Fig.~\ref{fig:pipeline} shows the goal is to (i) ensure accurate predictions, (ii) encourage sensor sparsity, (iii) maintain cross-modal alignment, and (iv) preserve temporal predictability. Formally, the total loss is defined as $\mathcal{L}_{\text{total}} =
\lambda_1 \mathcal{L}_{\text{task}} +
\lambda_2 \mathcal{L}_{\text{gating}} +
\lambda_3 \mathcal{L}_{\text{contrastive}} +
\lambda_4 \mathcal{L}_{\text{predictive}}$, where each $\lambda_i$ is a balancing coefficient controlling the contribution of the corresponding term.

\noindent
\textbf{Task Loss}: For accurate medical condition or activity classification based on multimodal input, we employ the standard cross‐entropy loss: $\mathcal{L}_{\text{task}} = - \frac{1}{N} \sum_{i=1}^{N} y_i \log \hat{y}_i$, where $y_i$ and $\hat{y}_i$ denote the ground‐truth and predicted class distributions, respectively.

\noindent
\textbf{Gating Regularization}: To encourage the AMC to minimize unnecessary sensor activations and thus improve energy efficiency, we introduce a sparsity‐inducing regularization term: $\mathcal{L}_{\text{gating}} = \frac{1}{M} \sum_{m=1}^M p_{\text{soft}}^{(m)}$, where $p_{\text{soft}}^{(m)}$ is the continuous gating probability for modality $m$. Intuitively, this term penalizes frequent activations, pushing the AMC to learn an optimal trade‐off between predictive performance and sensing cost. During backpropagation, this regularization interacts with the task loss to teach the AMC when a sensor is worth activating.

\noindent
\textbf{Contrastive Alignment Loss}: Biomedical signals from different sensors often capture correlated physiological phenomena. To ensure cross‐modal consistency even under modality dropout, we apply a contrastive alignment loss using an InfoNCE formulation:
\begin{equation}
\mathcal{L}_{\text{contrastive}} =
\log
\frac{
\exp(\text{sim}(\mathbf{h}_i, \mathbf{h}_j)/\tau)
}{
\sum_{k \neq i} \exp(\text{sim}(\mathbf{h}_i, \mathbf{h}_k)/\tau)
},
\end{equation}
where $\text{sim}(\cdot,\cdot)$ denotes cosine similarity and $\tau$ is a temperature hyperparameter. By pulling synchronized modalities closer in latent space and pushing unrelated ones apart, this term improves robustness to missing or gated inputs.

\noindent
\textbf{Predictive Coding Loss}: 
Our predictive coding term improves robustness when future sensory inputs are masked. Since the AMC may deactivate sensors to save energy, the model must preserve predictive capability by internally forecasting future multimodal states from current representations. We define the predictive loss as $\mathcal{L}_{\text{predictive}} = \left\|\text{MLP}(\mathbf{h}_t) - \mathbf{h}_{t+\delta}\right\|_2^2$, where $\mathbf{h}_t$ is the fused latent representation at time $t$, $\mathbf{h}_{t+\delta}$ is the target future embedding, $\delta$ is a small temporal offset, and $\text{MLP}(\cdot)$ is the predictor network.

\vspace{-2mm}
\subsection{Theoretical Analysis}

Consider $M$ modalities with information gains $I_1 \geq I_2 \geq ... \geq I_M$ relative to target $y$. Let $k^*(\epsilon) = |\{m : I_m \geq \epsilon\}|$ be the number of modalities needed for error at most $\epsilon$, and $T$ be the time horizon.
\vspace{-2mm}
\noindent
\newtheorem*{theorem*}{Theorem}
\begin{theorem*}[Sample Complexity]
To achieve prediction error $\leq \epsilon$, the required samples are:
\begin{itemize}[leftmargin=8pt]
    \item \textbf{Decoupled:} $\mathcal{O}(M \cdot T/\epsilon^2)$ samples
    \item \textbf{Joint (Ours):} $\mathcal{O}(k^*(\epsilon) \cdot T/\epsilon^2)$ samples in $\mathcal{O}(\log(M/\epsilon))$ rounds
\end{itemize}
\end{theorem*}
\vspace{-2mm}
\noindent
\textit{Proof.} Decoupled methods require $\mathcal{O}(1/\epsilon^2)$ samples per modality to estimate $\epsilon$-level contributions with high probability, totaling $\mathcal{O}(M/\epsilon^2)$ samples. Joint optimization leverages gradient $\nabla_\phi \mathcal{L}_{task}$ providing simultaneous estimates for all modalities. Using successive elimination with doubling batch sizes, round $r$ uses $2^r$ samples achieving confidence width $\mathcal{O}(1/\sqrt{2^r})$. Reaching $\epsilon$-accuracy needs $r=\mathcal{O}(\log(1/\epsilon^2))=\mathcal{O}(\log(1/\epsilon))$ rounds, with only $k^*$ modalities surviving elimination. This yields factor $M/k^*$ sample reduction and logarithmic convergence.

The advantage is most significant when $k^*(\epsilon){\ll}M$ (sparse sufficient statistics). Our results confirm this: MHEALTH has $k^*/M{\approx}0.38$ (high sparsity), while HMC has $k^*/M{\approx}0.69$ (low sparsity).

\vspace{-2mm}
\section{Experiments}
\subsection{Setup}
\textbf{Datasets}: We evaluate AMI on three public multimodal datasets.  
\textit{MHEALTH} \cite{mhealth_319} provides 12 activity classes from 3 wearable devices (ACC, GYRO, MAG, ECG), sampled at 50\,Hz and segmented into 2\,s windows. 
\textit{HMC Sleep}~\cite{alvarez_estevez_2022_hmc_sleep} includes EEG, ECG, PPG, EOG, and EMG overnight recordings; all channels are resampled to 100\,Hz and divided into standard 30\,s epochs.  
\textit{WESAD}~\cite{schmidt2018introducing} contains multimodal physiological signals (e.g., ECG, EDA, respiration, ACC) from chest and wrist devices; we resample to 100\,Hz and use 15\,s windows aligned with stress labels in 3 classes.

\noindent
\textbf{Models and Training Settings}: We use a unified hyperparameter setup across all experiments. Each modality is tokenized with specific patch size(10 for MHEALTH, 30 for HMC, 15 for WESAD), with a model embedding dimension of 256. The multimodal transformer backbone has 4 layers, 8 attention heads, and a feed-forward size of 1024, operating over a temporal history of 10 windows. The AMC is a lightweight MLP with a hidden size of 256. For temporal efficiency, we set $k_{\text{skip}}=2$ and a initial threshold $\theta=0.1$ for all modalities. Models are trained for 100 epochs with batch size 32, learning rate $1\times10^{-4}$ (cosine decay), and weight decay $1\times10^{-3}$. Since the AMC and FMPM are trained jointly, we apply BPTT with a gradient window of 10 steps. The total loss is a weighted sum of task ($\lambda_1{=}1.0$), gating ($\lambda_2{=}0.1$), contrastive ($\lambda_3{=}0.05$), and predictive ($\lambda_4{=}0.2$) objectives.

\begin{table}[!t]

\centering
\fontsize{8pt}{8pt}\selectfont
\setlength{\tabcolsep}{2pt}
\vspace{-3mm}
\caption{\small Performance comparison of the proposed method and prior state-of-the-art approaches on MHEALTH, HMC, and WESAD datasets. Accuracy (Acc.), F1 score (F1), and average modality sensing rate (Sensing) are reported. We also show the effect of varying the gating loss coefficient $\lambda_2$ on the trade-off between predictive performance and sensor efficiency.}

\vspace{-4mm}
\begin{tabular}{llccc}         
\toprule
Dataset & Method & Acc. (\%) & F1 & Sensing (\%) \\
\midrule
\multirow{5}{*}{\strut MHEALTH} 
        & Almujally et al. (2025) ~\cite{almujally2025wearable}& 94.67  & -  & 100  \\
        & Debache et al. ~\cite{debache2020lean}  & 98.2  &-  &100  \\
        & Sharma et al. ~\cite{sharma2023hybrid}& 99.07 & 99.1 & 100 \\
        \rowcolor{gray!30}
        &AMI (no sensing reduction) &99.10 &99.13 &100 \\
        \rowcolor{gray!30}
        &AMI ($\lambda_2=0.05$) &99.12 &99.12 &38.19 \\
        \rowcolor{gray!30}
        &AMI ($\lambda_2=0.1$) &99.04  &99.04 &33.3 \\
        \rowcolor{gray!30}
        &AMI ($\lambda_2=0.2$) &92.89  &91.78 &26.5 \\
\midrule
\multirow{5}{*}{\strut HMC}     
        & Estevez et al. (2022)~\cite{alvarez2021inter} &79.0  &-  & 100 \\
        & SleepMG~\cite{ma2024sleepmg} &69.24  &71.68  & 100 \\
        &PhysioOmni~\cite{jiang2025towards} &73.77  &77.79  &100  \\
        \rowcolor{gray!30}
        &AMI (no sensing reduction) &79.89 &79.12 &100 \\
        \rowcolor{gray!30}
        &AMI ($\lambda_2=0.001$) &78.21  &77.91  &73.01  \\
        \rowcolor{gray!30}
        &AMI ($\lambda_2=0.005$) &76.81  &77.21  &69.1  \\
        \rowcolor{gray!30}
        &AMI ($\lambda_2=0.01$) &67.43  &68.5  &53.54  \\
\midrule
\multirow{3}{*}{\strut WESAD}
        & Schmidt et al.~\cite{schmidt2018introducing}& 79.57 &  68.85 & 100 \\
        & Aleem et al.~\cite{abd2024deep}& 90.2 & 90.0 & 100 \\
        \rowcolor{gray!30}
        &AMI (no sensing reduction) &91.4 &90.02 &100 \\
        \rowcolor{gray!30}
        &AMI ($\lambda_2=0.05$) & 91.4 & 89.61 & 42.6 \\
        \rowcolor{gray!30}
        &AMI ($\lambda_2=0.1$) & 89.41 & 88.21 & 37.2 \\
\bottomrule
\end{tabular}
\vspace{-6mm}
\label{tab:results}
\end{table}

\vspace{-1mm}
\subsection{Accuracy Results}
We compare AMI with several recent SOTA approaches on three multimodal physiological datasets. \textit{On MHEALTH}, our model achieves 99.12\% accuracy and 99.12 F1, outperforming strong prior baselines such as Almujally et al.~\cite{almujally2025wearable} (94.67\%) and Debache et al.~\cite{debache2020lean} (98.2\%). Notably, this result is attained while activating only 38.19\% of sensing on average, showing that the agentic modality-selection mechanism can aggressively skip redundant samples without harming predictive performance. Even under stronger regularization ($\lambda_2 = 0.1$), the model retains competitive accuracy (99.04\%) while further lowering modality sensing to 33.3\%, demonstrating a stable efficiency–accuracy trade-off. \textit{On the HMC sleep-staging dataset}—a more challenging, long-horizon temporal task(where we take 30 seconds as our predict window length)—AMI achieves 78.21\% accuracy and 77.91 F1, closely matching or exceeding prominent supervised baselines~\cite{alvarez2021inter, ma2024sleepmg, jiang2025towards}. While previous fully-supervised models rely on full multimodal streams, our model requires only 73.01\% of sensing on average. This highlights the ability of the AMC to selectively activate modalities that are informative for sleep-stage transitions. \textit{On WESAD}, AMI achieves 91.4\% accuracy and 89.61 F1, surpassing recent deep multimodal architectures such as Aleem et al.~\cite{abd2024deep} (90.2\%), using only 42.6\% of modalities. Averaging across the three datasets, AMI reduces sensor usage by $48.8\%$ while increasing the SOTA accuracy by $1.9\%$.

We also show the trade-off between sensing rate and performance when varying the gating coefficient $\lambda_2$. For MHEALTH, a small weight (0.05) results in 99.23\% accuracy but limited reduction in sensing (38.19\%). Increasing $\lambda_2$ strengthens sensing sparsity, reaching 33.3\% modality usage at $\lambda_2=0.1$ with negligible accuracy change (99.04\%). A larger value (0.2) further reduces sensing to 26.5\%, but at the cost of performance degradation to 92.89\%.

Visualization heatmaps on MHEALTH and HMC (Fig.~\ref{fig:activation}) further demonstrate that the AMI and $\Sigma$-$\Delta$ Sensing effectively learns modality-dependent activation patterns.

\vspace{-1mm}

\subsection{Ablation Studies}


To better understand the contribution of each component in our architecture, we perform ablation experiments on the MHEALTH, shown in Table~\ref{tab:ablation}.  

\noindent
\textbf{Impact of AMC and $\Sigma$-$\Delta$ Sensing}: 
With only $\Sigma$-$\Delta$ Sensing enabled, the sensing rate increases to 78.90\% while achieving 99.09\% accuracy.
Conversely, using only AMC results in a sensing rate of 41.67\% with 98.89\% accuracy.
These results show that both mechanisms reduce sensing: AMC handles coarse-grained modality selection, while $\Sigma$–$\Delta$ suppresses redundant intra-window measurements, making the two strategies complementary.

\noindent
\textbf{Impact of components in FMPM}: 
Removing cross-modal fusion leads to notable performance degradation gives 98.11\% accuracy, confirming its function for learning joint multimodal representations.
Eliminating contextual temporal encoding reduces accuracy to 93.21\% (a 5.91\% drop), showing that long-range temporal context is critical, especially for physiological dynamics.
Removing both components results in a further decrease to 92.80\%.

\noindent
\textbf{Loss-function ablations}: Removing the contrastive consistency loss drops accuracy from 99.12\% to 97.61\% (1.51\%), indicating reduced robustness to modality corruption. Removing predictive coding yields 98.64\% (0.48\% drop), showing its role in stabilizing AMC decisions. Disabling both losses further reduces accuracy to 97.21\%.


\begin{figure}
    \centering
\includegraphics[width=1\linewidth]{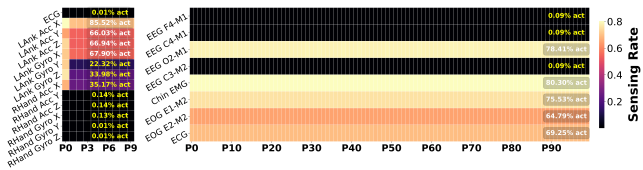}
\vspace{-8mm}
    \caption{\small Sensing rate heatmap over patches obtained from the proposed method on (\textit{Left}) MHEALTH and (\textit{Right}) HMC.}
    \label{fig:activation}
    \vspace{-4mm}
\end{figure}

\begin{table}[!t]
\centering
\fontsize{8pt}{8pt}\selectfont
\setlength{\tabcolsep}{3pt}
\caption{\small Ablation study on MHEALTH. We report metrics for the AMI pipeline (all components turned on) and for variants with individual components removed. Sensing reduction, FMPM, and loss-function contributions are evaluated separately to isolate their effects.}

\vspace{-4mm}
\begin{tabular}{lccc}
\toprule
Configuration & Acc. & F1 & Sensing \\
\midrule
AMI (full model) &99.12 &99.12 &38.19  \\
\midrule
\textit{Sensing reduction ablation} \\
w/o AMC &98.89  &99.02 & 78.90 \\
w/o $\Sigma$-$\Delta$ Sensing &99.09 &99.12 &41.67  \\
\midrule
\multicolumn{4}{l}{\textit{FMPM ablation (w/o any sensing reduction)}} \\
w/o cross-modal fusion &98.11  &98.12 &-  \\
w/o context encoding &93.21  &92.17 &-  \\
w/o cross-modal fusion and context encoding  &92.80  &91.45 &- \\
\midrule
\textit{loss ablation (w/o any sensing reduction)} \\
w/o contrastive loss &97.61 &97.84 &-\\
w/o predictive coding loss &98.64 & 98.21 &-\\
w/o contrastive loss and predictive coding loss &97.21 &97.41 &-\\

\bottomrule
\end{tabular}
\vspace{-3mm}
\label{tab:ablation}
\end{table}

\vspace{-1mm}
\subsection{Robustness Studies}
\noindent
\textbf{Random modality masking}: The AMC is disabled and we use a pretrained model by randomly dropping modalities according to a Bernoulli mask with probability $p$, simulating unexpected sensor failures or severe signal corruption. This setup isolates the robustness of the FMPM, as predictions rely solely on the learned multimodal representations without adaptive scheduling. As shown in Table~\ref{tab:masking}, at $p=0.2$, accuracy remains at 99.10\%. Even when 50\% of modalities are randomly removed ($p=0.5$), accuracy only drops slightly to 99.08\%. Significant degradation occurs for high masking rates ($p=0.8$). 


\begin{table}[!t]
\centering
\fontsize{8pt}{8pt}\selectfont
\caption{\small \textit{(Left)} Robustness under random modality dropout, and \textit{(Right)} Effect of input sampling rate on MHEALTH.}
\label{tab:robustness}
\vspace{-4mm}
\begin{minipage}{0.4\columnwidth}
\centering
\setlength{\tabcolsep}{4pt}
\begin{tabular}{ccc}
\toprule
$p$ & Acc. (\%) & F1 (\%) \\
\midrule
0     & 99.12 & 99.12 \\
0.2   & 99.10 & 99.12 \\
0.5   & 99.08 & 99.01 \\
0.8   & 87.21 & 88.47 \\
\bottomrule
\end{tabular}
\vspace{2mm}
\captionsetup{font=small}
\label{tab:masking}
\end{minipage}
\begin{minipage}{0.55\columnwidth}
\centering
\setlength{\tabcolsep}{1pt}
\begin{tabular}{cccc}
\toprule
Sampling rate     & Acc. (\%) & F1 (\%) & Sensing (\%) \\
\midrule
50Hz (default) & 99.12 & 99.12 & 38.19 \\
25Hz           & 94.30 & 94.18 & 39.0  \\
5Hz            & 87.04 & 86.85 & 39.0  \\
\bottomrule
\end{tabular}
\vspace{2mm}
\captionsetup{font=small}
\label{tab:sampling}
\end{minipage}
\vspace{-8mm}
\end{table}

\noindent
\textbf{Low sampling rate evaluation}: We further assess robustness to temporal signal degradation on MHEALTH by resampling all sensor streams to lower sampling rates. Since reducing the sampling frequency reduces the number of samples per window, we proportionally adjust the patch size to preserve consistent temporal coverage per patch, specifically for any new sampling rate $f$ we have $P(f) = P_0 \cdot \frac{f}{f_0}$, where $f_0$ is the default sampling rate (50 Hz in MHEALTH) and $P_0$ is the base patch size. We observe that the model maintains performance at 25~Hz with $4.8\%$ accuracy degradation, and remains functional even at 5~Hz with an acc. of $87.04\%$. 

\vspace{-2mm}
\subsection{Hardware Analysis}

We benchmark our system across the HMC, MHEALTH, and WESAD datasets on four representative hardware platforms:
an ARM CPU (Apple~M1),
NVIDIA Jetson Orin (TensorRT),
RTX A6000 (PyTorch), and
RTX A6000 (TensorRT).
For each dataset–platform pair, we report per-iteration latency and energy for:
(i) the \emph{AMC},
(ii) the \emph{FMPM},
and (iii) the \emph{AMI} pipeline, consisting of \emph{AMC} and \emph{FMPM}.

\noindent\textbf{Setup.}
All latency and energy measurements use a batch size of $B{=}1$, with $K{=}10$ warm-up iterations and $N{=}100$ timed iterations.
GPU energy is measured from the NVML total-energy counter.
TensorRT deployments on Jetson and A6000 apply standard optimizations such as layer fusion, kernel autotuning, and FP16 execution.

\noindent\textbf{Latency \& Energy Savings}:
To evaluate the runtime impact of AMC and $\Sigma$-$\Delta$ Sensing, we compare AMI with an FMPM-only baseline under the \emph{same} sensing configuration.
Table~\ref{tab:trt_latency_energy} shows that on the RTX~A6000 (TensorRT), among all three datasets. AMI achieves its strongest gains on WESAD, reducing latency by 56.33\% and energy by 37.1\% while sensing only 42\% of the modalities. 
On average across datasets, AMI delivers 31.9\% lower latency and 24.8\% lower energy consumption, which directly translates to longer battery life in continuous monitoring scenarios.


\vspace{1mm}
\begin{table}[!t]
\centering
\fontsize{8pt}{8pt}\selectfont
\setlength{\tabcolsep}{1.3pt}
\caption{\small Runtime latency and energy comparison between the baseline and proposed AMI during inference on an ARM CPU.}
\vspace{-3mm}
\begin{tabular}{llccc}
\toprule
Dataset & Model & Lat. (ms) & Eng.\,(mJ) & Savings (lat./eng.) \\
\midrule
\multirow{2}{*}{\strut MHEALTH} 
&\small{FMPM (100\% sensing)} & 5.28 & 254.74 & -- \\
&\small{AMI (38\% sensing)}        & 3.69 & 200.14 &  30.01\% / 21.40\% \\
\hline
\multirow{2}{*}{\strut HMC} 
&\small{FMPM (100\% sensing)} & 3.17 & 175.39 & -- \\
&\small{AMI (73\% sensing)}        & 2.87 & 147.31 & 9.46\% / 16.01\% \\
\hline
\multirow{2}{*}{\strut WESAD} 
&\small{FMPM (100\% sensing)} & 5.68 & 257.32 & -- \\
&\small{AMI (42\% sensing)}   & 2.48 & 161.64 &  56.33\% / 37.1\% \\
\bottomrule
\end{tabular}
\vspace{-2mm}
\label{tab:trt_latency_energy}
\end{table}

Fig.~\ref{fig:hw_latency_energy} summarizes per-iteration latency and energy across the four hardware platforms as the number of modalities varies.
On the left, latency is decomposed into AMC and FMPM. Across all platforms, the FMPM dominates the end-to-end latency, while the AMC branch remains consistently much smaller,
reinforcing that the AMC overhead is negligible relative to FMPM cost and substantially smaller than the gains obtained by reducing sensing. On the right, we report per-iteration energy on the GPU platforms.
TensorRT significantly accelerates the AMI: Across sensing-rate settings, A6000 (TensorRT) achieves a 5.1–5.4$\times$ latency speedup over A6000 (PyTorch) and reduces GPU energy consumption by 68.2–75.1\%. Comparing hardware platforms, Jetson (TensorRT) provides a 49.3\% energy reduction, while being 36.4\% slower in latency than A6000 (TensorRT) on average across sensing rates.

\begin{figure}[t]
\centering
\vspace{-2mm}
\includegraphics[width=0.98\linewidth]{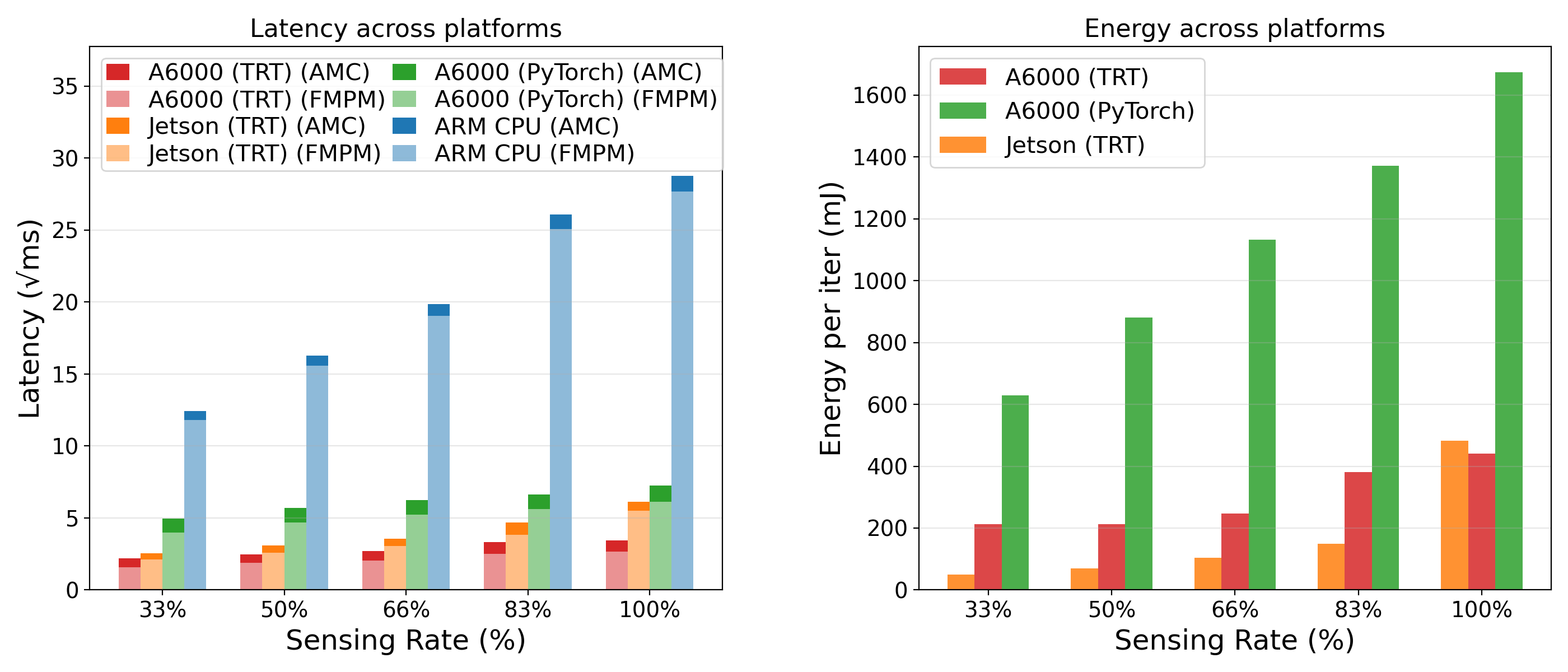}
\vspace{-4mm}
\caption{\small Per-iteration latency (left) and energy (right) as the sensing rate varies on MHEALTH. Measurements are across ARM CPU, Jetson (TensorRT) and A6000 (PyTorch and TensorRT). Latency is decomposed into AMC and FMPM and values are in $\sqrt{}$ of ms. 
}
\label{fig:hw_latency_energy}
\vspace{-5mm}
\end{figure}

\vspace{-1mm}

\section{Conclusions}

This paper presented a unified framework that jointly learns when to sense and how to infer, addressing the core energy–accuracy bottlenecks of multimodal medical monitoring. By integrating a learned agent, $\Sigma$–$\Delta$ temporal sensing, and a foundation-backed prediction model, our framework enables dynamic sensing policies that adapt to both task demands and signal redundancy. Our multi-objective training aligns sensing, fusion, and temporal prediction in a single end-to-end system, yielding substantial reductions in sensing cost while improving the SOTA accuracy. Experiments across three diverse biomedical datasets show 31.9\% lower latency and 24.8\% lower energy consumption. These results show that joint sensing–inference optimization is a principled and practical path toward efficient edge intelligence in wearable and implantable systems. By enabling long-duration, low-power multimodal monitoring with reliable on-device inference, this work offers a promising foundation for continuous clinical assessment and earlier intervention in real-world patient care.


\vspace{-0.05in}
\bibliographystyle{ACM-Reference-Format}
\bibliography{main}

\end{document}